\documentclass[a4paper, 12pt, oneside]{article}
\usepackage{amsmath}
\usepackage{amssymb}

\usepackage[amsmath,thmmarks]{ntheorem}
\newcounter{thm} 

\newtheorem{Thm}[thm]{Theorem}

\newtheorem{Lem}[thm]{Lemma}
\newtheorem{Prop}[thm]{Proposition}
\newtheorem{Cor}[thm]{Corollary}

\newtheorem{Conj}[thm]{Conjecture}

\theoremstyle{nonumberplain}
\theoremheaderfont{\normalfont\itshape}
\theorembodyfont{\normalfont}
\theoremseparator{.}
\theoremsymbol{\ensuremath{\square}}
\newtheorem{proof}{Proof}

\def \Z {\mathbb Z}

\def \d {\mathrm{d}}
\def \tr {\mathrm{Tr}}

\def \csch {\operatorname{csch}}

\begin{document}
\title{From Kontsevich-Witten to linear Hodge integrals via Virasoro operators}
\author{Gehao Wang}
\date{}
\maketitle

\abstract
We give a proof of Alexandrov's conjecture on a formula connecting the Kontsevich-Witten and Hodge tau-functions using only the Virasoro operators. This formula has been confirmed up to an unknown constant factor. In this paper, we show that this factor is indeed equal to one by investigating series expansions for the Lambert W function on different points. 

{\bf Keywords:} Lambert W function, tau-functions, Virasoro operators.

\section{Introduction}
The KP (Kadomtsev-Petviashvili) hierarchy is a completely integrable system of partial differential equations for a formal power series $F$ in variables $\{q_1,q_2,\dots\}$. It is a generalization of the KdV hierarchy, which can be viewed as a 2-reduced KP hierarchy. The exponent $\exp(F)$ for any solution is usually called the tau-function for the hierarchy. It is well-known that the set of all tau-functions for the KP hierarchy forms an orbit under the action of the so-called $\widehat{GL(\infty)}$ group. This group is constructed via the exponential map from the infinite dimensional Lie algebra $\widehat{\mathfrak{gl}(\infty)}$ elements (see, e.g.,\cite{DJKM},\cite{MJE}), and the action is interpreted as applying the exponential of differential operators belonging to this Lie algebra on tau-functions.

Among the fundamental operators from the algebra $\widehat{\mathfrak{gl}(\infty)}$, the Heisenberg and Virasoro operators are widely used when we construct relations between tau-functions. In this context, the Heisenberg operators $\alpha_n$ are in the form
\[
\alpha_n=
\begin{cases}
q_{-n} & n<0 \\
0 &  n=0 \\
n\frac{\partial}{\partial q_n} &  n>0.
\end{cases}
\]
and they span the Heisenberg algebra. The Virasoro operators $L_m, m\in \Z,$ are in the form
$$L_m=\sum_{k>0,k+m>0} (k+m)q_k\frac{\partial}{\partial q_{k+m}}+\frac{1}{2}\sum_{a+b=m} ab\frac{\partial^2}{\partial q_a\partial q_b}+\frac{1}{2}\sum_{i,j>0,i+j=-m}q_iq_j,$$
and they form a representation of the Virasoro algebra. In terms of their applications on connecting tau-functions, for example, Kazarian established the relation between the Hurwitz and Hodge tau-function in \cite{MK} using operators $L_m$ and $\alpha_n$ with $m,n<0$. Then, through the Hodge tau-function, this relation has been extended to the Kontsevich-Witten tau-function in \cite{AE} and \cite{LW} using $L_m$ and $\alpha_n$ with $m,n>0$, and furthermore to the partition function of r-spin numbers in \cite{DLM}. Recently, a formula connecting the Kontsevich-Witten and Brezin-Gross-Witten tau-function has been discussed in \cite{W1} using the $\widehat{\mathfrak{sl}_2}$ operators $\alpha_{2k+1}$ and $\widetilde{L}_{2m}$, which form a sub-algebra of $\widehat{\mathfrak{gl}(\infty)}$ that preserves the KdV hierarchy, where $\widetilde{L}_{2m}$ denotes the odd variable part of $L_{2m}$. Formulas involving such operators are particularly interesting not only because they preserve the KP (KdV) integrability, but also they have some practical properties themselves. The exponential of $\partial q_n$ behaves as a shift, and the exponential of the derivation part of $L_m$ behaves as a change of variables. This gives us a convenient way to describe what the action of such operator does on functions. In addition, operators $\alpha_n$ and $L_m$ usually appear in some nice commutator relations like
\begin{equation*}
[\frac{1}{n}\alpha_n,L_k]=\alpha_{n+k}, \quad, [L_m,L_n] =(m-n)L_{m+n}+\frac{1}{12}(m^3-m)\delta_{m+n,0}.
\end{equation*}
As one of the many examples, in \cite{GW}, these commutator relations are the crucial part in the proof of the equivalence relations among the Virasoro constraints, cut-and-join equation and polynomial recursion relation for the linear Hodge integrals using the main formula in \cite{LW}. 

In the paper \cite{A}, Alexandrov posted a conjecture stating that there exists a $\widehat{GL(\infty)}$ operator consisting of only the Virasoro operators that connects the Hodge tau-function $\tau_{Hodge}$ and Kontsevich tau-function $\tau_{KW}$. Later, in \cite{AE}, he reformulated this conjecture into a conjectural formula as the following.
\begin{Conj}[Alexandrov, \cite{AE}]\label{Conj1}
\begin{equation}\label{Conj1eq}
\tau_{KW}=\widehat{G}_{+}\tau_{Hodge},
\end{equation}
where
$$\widehat{G}_{+}=\beta^{-\frac{4}{3}L_0}\exp(\sum_{k=1}^{\infty}\widehat{a}_k\beta^{-k}L_k )\beta^{\frac{4}{3}L_0}.$$
The coefficients $\{\widehat{a}_k\}$ are determined by a series $f_{+}$,
\begin{equation}\label{def:f+}
f_{+}=\exp(\sum_{k=1}^{\infty} \widehat{a}_kz^{1-k}\frac{\partial}{\partial z})\cdot z,
\end{equation}
and $f_{+}$ is given as the solution of the following equation
\begin{equation}\label{f+}
\frac{f_{+}(z)}{1+f_{+}(z)}\exp(-\frac{f_{+}(z)}{1+f_{+}(z)})=E\exp(-E),
\end{equation}
where
$$E=1+\sqrt{(\frac{1}{1+f_{+}(z)})^2+\frac{4}{3z^3}}.$$
\end{Conj}
Alexandrov verified this formula up to a constant factor $C(\beta)$ with
$C(\beta)=1+\sum_{k=1}^{\infty}c_k\beta^{2k}$, and checked by direct computation that, at least, $C(\beta)=1+\mathcal{O}(\beta^{10})$. Comparatively, in \cite{LW}, the authors used a completely different approach to establish an explicit formula connecting these two tau-functions using Heisenberg and Virasoro operators, (see Theorem \ref{LWthm}), and then transformed this formula to another one involving only Virasoro operators as a corollary (see Eq.\eqref{emeq}). This corollary proves the conjecture in \cite{A}. However, as mentioned in \cite{LW}, it was not clear what the relation between the two formulas Eq.\eqref{Conj1eq} and Eq.\eqref{emeq} is.

The goal of this paper is to confirm that these two formulas are equivalent, and conclude that
\begin{Thm}
Conjecture \ref{Conj1} is true, that is, $C(\beta)=1$.
\end{Thm}
Our method starts from studying the power series used to describe the corresponding formula. In fact, the series $f_{+}$ in Conjecture \ref{Conj1} and the series $\theta(f(z))$ used to describe the differential operators in Eq.\eqref{emeq} are both closely related to the {\bf Lambert W function}. As we will see later in our context, despite their complete different definitions, these two series are in fact inverses of each other. To show this, we need to investigate their relations to the Lambert W function $W(z)$. As a multi-valued function, the difference between two branches of $W(z)$ is actually nontrivial, even for real-valued $z$. This property of $W(z)$ plays a very important role in determining the relation between the two series $f_{+}$ and $\theta(f(z))$.

The present paper is organized as follows. In Sect.\ref{S2} we review the Kontsevich-Witten and Hodge tau-functions and some formulas connecting them. In Sect.\ref{S3} we introduce several series expansions for the Lambert W function and their properties. In Sect.\ref{S4}, we study the connection between the two series $f_{+}$ and $\theta(f(z))$, and prove Conjecture \ref{Conj1}.

\section{Formulas connecting the two tau-functions}\label{S2}
\subsection{The tau-functions}
Let $\overline{M}_{g,n}$ be the Deligne-Mumford compactification of the moduli space of complex stable curves of genus $g$ with $n$ marked points, and $\psi_i$ be the first Chern class of the cotangent line over $\overline{M}_{g,n}$ at the $i$th marked point. The intersections of the $\psi$-classes are evaluated by the integral:
$$<\tau_{d_1}\dots \tau_{d_n}>=\int_{\overline{M}_{g,n}}\psi_1^{d_1}\dots \psi_n^{d_n}.$$
The Kontsevich-Witten generating function (\cite{K},\cite{WE}) is defined as
$$F_K(t)=\sum <\tau_0^{k_0}\tau_1^{k_1}\dots>\frac{t_0^{k_0}}{k_0!}\frac{t_1^{k_1}}{k_1!}\dots.$$
And Kontsevich matrix model is the following matrix integral over the space of Hermitian matrices $\Phi$:
\begin{equation*}
Z_K=\frac{\int [\d\Phi] \exp\left(-\tr (\frac{\Phi^3}{6}+\frac{\Lambda\Phi^2}{2}) \right)}{\int [\d\Phi] \exp\left(-\tr \frac{\Lambda\Phi^2}{2} \right)},
\end{equation*}
where $\Lambda$ is the diagonal matrix. It gives us a representation $Z_K=\exp(F_K(t))$ under the Miwa parametrization
$$t_k=(2k-1)!!\tr \Lambda^{-2k-1}.$$
The Kontsevich-Witten tau-function is defined to be $\exp{(F_K(q))}$, where
$$F_K(q)=\left.F_K(t)\right\vert_{t_k=(2k-1)!!q_{2k+1}}.$$
It is well-known that $\exp{(F_K(q))}$ is a tau-function for the KdV hierarchy \cite{K}. 

Let $\lambda_j$ be the $j$th Chern class of the Hodge bundle over $\overline{M}_{g,n}$ whose fibers over each curve is the space of holomorphic one-forms on that curve. The linear Hodge integrals are the intersection numbers of the form
$$<\lambda_j\tau_{d_1}\dots \tau_{d_n}>=\int_{\overline{M}_{g,n}}\lambda_j\psi_1^{d_1}\dots \psi_n^{d_n}.$$
They are trivial when the numbers $j$ and $d_i$ do not satisfy the condition
\begin{equation*}
j+\sum_{i=1}^n d_i=dim(\overline{M}_{g,n})=3g-3+n.
\end{equation*}
The linear Hodge partition function is defined as
$$F_H(u,t)=\sum (-1)^j <\lambda_j \tau_0^{k_0}\tau_1^{k_1}\dots>u^{2j} \frac{t_0^{k_0}}{k_0!}\frac{t_1^{k_1}}{k_1!}\dots$$
where $u$ is the parameter marking the $\lambda$-class. It is easy to see that $F_H(0,t)=F_K(t)$. For $k\geq 0$, let
\begin{align*}
& D=(u+z)^2z\frac{\partial}{\partial z}, \quad D^{k}z=\sum_{j=1}^{2k+1} \alpha_{j}^{(k)}u^{2k+1-j}z^j,\\
& \widetilde{\phi_k}(u,q)=\sum_{j=1}^{2k+1} \alpha_{j}^{(k)} u^{2k+1-j}q_j.
\end{align*}
The Hodge tau-function is defined to be $\exp(F_H(u,q))$, where
$$F_H(u,q)=\left.F_H(u,t)\right\vert_{t_k=\widetilde{\phi_k}(u,q)},$$
and it is a tau-function for the KP hierarchy \cite{MK}.

\subsection{The formulas}
Here we review some known formulas that use differential operators to connect the two tau-functions. Using Mumford's theorem in \cite{Mu}, one can recover $F_H(u,t)$ from $F_K(t)$ using the following formula (cf.\cite{FP},\cite{AG1})
\begin{equation}\label{1}
\exp(F_H(u,t))=e^W\cdot \exp(F_K(t)),
\end{equation}
where
\begin{equation}\label{W}
W=-\sum_{k\geq 1} \frac{B_{2k}u^{2(2k-1)}}{2k(2k-1)}(\frac{\partial}{\partial t_{2k}}-\sum_{i\geq 0} t_i\frac{\partial}{\partial t_{i+2k-1}}+\frac{1}{2}\sum_{i+j=2k-2} (-1)^i\frac{\partial^2}{\partial t_i\partial t_j}).
\end{equation}
Here $B_{2k}$ is the Bernoulli numbers defined by:
$$\frac{t}{e^t-1}=\sum_{m=0}^{\infty} B_m\frac{t^m}{m!}.$$
However, the operator $W$ does not belong to the $\widehat{\mathfrak{gl}(\infty)}$ algebra. In the paper \cite{LW}, the authors introduced a method to decompose the operator $W$ into several factors using the Zassenhaus formula, and then transform the formula \eqref{1} into the following one that connects the two tau-functions $\exp(F_H(u,q))$ and $\exp(F_K(q))$ using the Virasoro and Heisenberg operators.
\begin{Thm}[Liu-Wang,\cite{LW}]\label{LWthm}
Define the series $f(z)$ to be
\begin{equation}\label{functionf}
f(z) = (-2\log(1-\frac{1}{1+z})-\frac{2}{1+z})^{-\frac{1}{2}}=z+\frac{2}{3}-\frac{1}{12}z^{-1}+\dots.
\end{equation}
Then the relation between $\exp(F_H(u,q))$ and $\exp(F_K(q))$ can be written as the following:
\begin{equation}\label{maineq}
\exp(F_H(u,q))=\exp\left(\sum_{m=1}^{\infty} a_mu^mL_m\right)\exp(P)\cdot\exp{(F_K(q))}.
\end{equation}
The coefficients $\{a_m\}$ can be computed from the equation
\begin{equation} \label{eqn:a}
\exp(\sum_{m>0}a_mz^{1-m}\frac{\partial}{\partial z})\cdot z=f(z).
\end{equation}
And $P$ is of the form
$$P=-\sum_{k=1}^{\infty} b_{2k+1}u^{2k}\frac{\partial}{\partial q_{2k+3}},$$
where the numbers $\{b_{2k+1}\}$ are uniquely determined by the recursion relation,
\begin{equation}\label{re}
(n+1)b_n=b_{n-1}-\sum_{k=2}^{n-1}kb_kb_{n+1-k}
\end{equation}
with $b_1=1, b_2=1/3$.
\end{Thm}
Furthermore, this theorem implies another formula that connects the two tau-functions using only the Virasoro operators
\begin{Cor}[\cite{LW}]\label{em}
Define the power series $\theta(z)$ to be
\begin{equation}\label{theta}
\theta(z)=\left(3\sum_{k=0}^{\infty} \frac{b_{2k+1}}{2k+3}z^{-2k-3}\right)^{-\frac{1}{3}}.
\end{equation}
Then
\begin{equation}\label{emeq}
\exp(F_H(u,q))=\exp\left(\sum_{m=1}^{\infty} e_mu^mL_m\right)\cdot\exp(F_K(q)),
\end{equation}
where the coefficients $\{e_m\}$ are determined by the equation 
\begin{equation}\label{def:em}
\exp(\sum_{m=1}^{\infty} e_{m}z^{1-m}\frac{\partial}{\partial z})\cdot z=\theta(f(z)).
\end{equation}
\end{Cor}
The above formula is obtained by first replacing the operator $P$ in Eq.\eqref{maineq} with the Virasoro operator $\sum_{m=1}^{\infty} -l_{m}u^{2m}L_{2m}$, where
$$\exp(-\sum_{m=1}^{\infty} l_{m}z^{1-2m}\frac{\partial}{\partial z})\cdot z=\theta(z).$$
This is due to the fact that the Kontsevich-Witten tau-function satisfies the Virasoro constraints
\begin{equation}\label{VforKW}
\left(L_{2m}-(2m+3)\frac{\partial}{\partial q_{2m+3}}\right)\cdot \exp(F_K(q))=0,
\end{equation}
and
\begin{align*}
&\exp\sum_{m=1}^{\infty} l_{m}\left(L_{2m}-(2m+3)\frac{\partial}{\partial q_{2m+3}}\right)\\
=&\exp(\sum_{m=1}^{\infty} l_{m}L_{2m})\exp(-\sum_{k=1}^{\infty} b_{2k+1}\frac{\partial}{\partial q_{2k+3}}).
\end{align*}
Then, since
$$\theta(f(z))=\exp(\sum_{m=1}^{\infty}a_mz^{1-m}\frac{\partial}{\partial z})\exp(-\sum_{m=1}^{\infty} l_{m}z^{1-2m}\frac{\partial}{\partial z})\cdot z,$$
we have
$$\exp\left(\sum_{m=1}^{\infty} e_mu^mL_m\right)=\exp\left(\sum_{m=1}^{\infty} a_mu^mL_m\right)\exp\left(\sum_{m=1}^{\infty} -l_{m}u^{2m}L_{2m}\right). $$

Using the Kac-Schwarz operators, Alexandrov established the relation Eq.\eqref{Conj1eq} up to an unknown factor $C$ in \cite{AE}. This relation can be written as the following formula connecting the two tau-functions $\exp(F_H(u,q))$ and $\exp(F_K(q))$ :
\begin{equation}\label{Alexconj}
\exp(F_K(q))=C(u) \exp\left(\sum_{k=1}^{\infty}\widehat{a}_ku^{k}L_k  \right) \cdot\exp(F_H(u,q)),
\end{equation}
where $u=\beta^{\frac{1}{3}}$. In fact, the function $f_{+}$ in Eq.\eqref{f+} is uniquely determined by the asymptotics for large $|z|$, and it can be expressed as the composition of $f_{+1}$ and $f_{+2}$
\begin{equation}\label{f+1+2}
f_{+}(z)=f_{+1}(f_{+2}(z)),
\end{equation}
where
\begin{equation}\label{f+1}
f_{+1}=\frac{1}{z\exp(z^{-1})\sinh(z^{-1})-1},
\end{equation}
and $f_{+2}$ satisfies the equation
\begin{equation}\label{f+2}
\frac{1}{(f_{+2})^2}\coth\left(\frac{1}{f_{+2}}\right)-\frac{1}{f_{+2}}=\frac{1}{3z^3}.
\end{equation}
Briefly speaking, Alexandrov's construction is based on the description of the tau-functions using Sato Grassmannian. By the Wick's theorem, one can obtain a set of basis vectors from the determinantal representation of the tau-function. This set of basis vectors defines a subspace $\mathcal{W}$ of an infinite dimensional Grassmanian. For the differential operator in Eq.\eqref{Alexconj}, let
\begin{equation*}
\widehat{V}_i=\sum_{k>0}\widehat{a}_k^{(i)}u^kL_k,
\end{equation*}
where
\begin{equation}\label{f+i}
\exp\left(\sum_{k>0}\widehat{a}_k^{(i)}z^{1+k}\frac{\d}{\d z}\right)\cdot z=f_{+i},
\end{equation}
then, we have
\begin{equation*}
\exp\left(\sum_{k=1}^{\infty}\widehat{a}_ku^{k}L_k  \right)=\exp(\widehat{V}_2)\exp(\widehat{V}_1).
\end{equation*}
Since the Kontsevich-Witten tau-function satisfies the Virasoro constraints \eqref{VforKW}, we have
\begin{equation*}
\exp(-\widehat{V}_2)\exp(F_K(q))=\exp(-\widehat{N})\exp(F_K(q)),
\end{equation*}
where
\begin{equation*}
\widehat{N}=\sum_{k\geq 2}\frac{2^{2k}B_{2k}}{(2k)!}u^{2k-2}\frac{(2k+1)\partial}{\partial q_{2k+1}}.
\end{equation*}
On the other hand, let
\begin{align*}
n(z)&=\sum_{k\geq 2}\frac{2^{2k}B_{2k}}{(2k)!}z^{-2k-1}=z^{-2}\coth(z^{-1})-z^{-1}-\frac{z^{-3}}{3},\\
v_1(z)&=\exp\left(\sum_{k>0}\widehat{a}_k^{(1)}u^kz^{1+k}\frac{\d}{\d z}\right)\cdot z,
\end{align*}
where $n(z)$ and $v_1(z)$ correspond to the operators $\widehat{N}$ and $\widehat{V}_1$ respectively. Then, for a suitable set of basis vectors $\{\Phi_k^{Hodge}\}$ of the Hodge tau-function, Alexandrov showed that the vectors
$$u^{k-1}\exp\left(u^{-3}n(u^{-1}z)\right)\exp(v_1)\Phi_k^{Hodge}, \quad k>0$$
all belong to the vector space of the Grassmannian for the Kontsevich-Witten tau-function. Since the action of the group element on the tau-function is equivalent to the action of corresponding operators from the algebra $w_{1+\infty}$ on the set of basis vectors, (see e.g., Section 1.4 in \cite{AE}), the two tau-functions should coincide up to a constant factor $C(u)$. This gives us Eq.\eqref{Alexconj}. But, at this stage, this approach can not determine the value of $C(u)$.

\section{Series for the Lambert W function}\label{S3}
In this section we review some properties of the Lambert W function and its related series. We refer the readers to \cite{CGHJK} and other related articles for further details. The Lambert W function $W(z)$ is defined to be the function satisfying
\begin{equation*}
W(z)e^{W(z)}=z.
\end{equation*}
Since the map $w\rightarrow we^w$ is not injective, the function $W(z)$ is multi-valued. For real-valued $W$, the function is defined in $z\geq -e^{-1}$. There is one branch for $z\geq 0$, and two branches in $-e^{-1}\leq z<0$. The requirement $W\geq -1$ gives us a single-valued function $W_0(z)$ in the domain $z\geq -e^{-1}$ with $W_0(-e^{-1})=-1, W_0(0)=0$ and $W_0(0)>0$ for $z>0$. The requirement $W\leq -1$ gives us the other branch $W_{-1}(z)$ in the domain $-e^{-1}\leq z<0$ with $W_{-1}(-e^{-1})=-1, W_{-1}(0^{-})=-\infty$. 

In \cite{L}, Lauwerier considered another independent variable 
\begin{equation}\label{def:p}
p=\frac{1}{2}(W_0(z)-W_{-1}(z)),
\end{equation}
which concerns the branch difference of $W(z)$ in the domain $-e^{-1}\leq z<0$. This definition immediately implies that
\begin{align}\label{W0p}
W_0(z)=-pe^p\csch p, \quad W_{-1}(z)=-pe^{-p}\csch p
\end{align}
with $p\geq 0$, and
\begin{equation}\label{eqn:p}
z=-p e^{-p\coth p}\csch p,
\end{equation}
where
\begin{equation*}
\coth p=\frac{\cosh p}{\sinh p}=\frac{e^{2p}+1}{e^{2p}-1},\quad \csch p=\frac{1}{\sinh p}=\frac{2e^p}{e^{2p}-1}
\end{equation*}
are the hyperbolic functions. 


In \cite{JK}, Karamata considered the solution $\mu$ of the equation
\begin{equation}\label{eqn:mu}
(1-x)e^{-(1-x)}=(1+\mu)e^{-(1+\mu)}.
\end{equation}
Clearly, $\mu=-x$ is one solution. In fact, the solutions can be expressed using the Lambert W function
\[
\mu=
\begin{cases}
-1-W_0(-(1-x)e^{-(1-x)}) \\
-1-W_{-1}(-(1-x)e^{-(1-x)}).
\end{cases}
\]
If $x>0$, then $W_0(-(1-x)e^{-(1-x)})=-1+x$, which gives us $\mu=-x$. And the expression for $W_{-1}$ is a power series. Suppose, in this case, $W_{-1}=-1-\mu$, where
$$\mu=\sum_{n=1}^{\infty}c_nx^n$$
with $c_1=1$. Then, after differentiating both sides of the Eq.\eqref{eqn:mu} with respect to $x$, we have
$$\mu\frac{\d \mu}{\d x}=(1+\mu)\frac{x}{x-1}$$
which can give us the recursion relation
\begin{equation}\label{re:c}
(n+1)c_n=2+\sum_{j=2}^{n-1}c_j(1-jc_{n-j+1})
\end{equation}
and $c_2=2/3$. If $x<0$, then it is $W_{-1}$ being $-1+x$ and $W_0$ being the power series defined above.

Next, we look at the following two series for the Lambert W function
\begin{equation*}
u=-W_0(-e^{-1-\frac{z^2}{2}}) \quad\mbox{and}\quad v=-W_{-1}(-e^{-1-\frac{z^2}{2}}).
\end{equation*}
The results from \cite{MM} indicate a connection between the above two series and the Stirling's approximation of $n!$. Let
\begin{equation}\label{v}
v=1+\sum_{i=1}^{\infty}b_iz^i
\end{equation}
with $z\geq 0$. If we differentiate the equation
\begin{equation*}
ve^{1-v}=e^{-\frac{1}{2}z^2}.
\end{equation*}
by $z$ on both sides, we have $v^{'}(v-1)=zv$, which gives us exactly the recursion relation \eqref{re} on the coefficients $\{b_i\}$. The series $u$ is another solution of the above equation with the form being
\begin{equation}\label{u}
u=1+\sum_{i=1}^{\infty} (-1)^ib_iz^i.
\end{equation}
Starting from the Euler integral of the second kind
\begin{equation*}
n!=\int_0^\infty  x^{n-1} e^{-x}\mathrm{d}x,
\end{equation*}
we can split the range of integration into
\begin{align*}
n!&=n^{n+1}e^{-n}\left\{\int_0^1 (ue^{1-u})^n \d u +\int_1^\infty (ve^{1-v})^n \d v\right\}\\
&=n^{n+1}e^{-n}\int_0^{\infty} ze^{-nz^2/2}\left( \frac{1}{1-v}-\frac{1}{1-u}\right)\d z\\
&=n^ne^{-n}\sqrt{2\pi n}\sum_{i=0}^{\infty}(2i+1)!!b_{2i+1}n^{-i}.
\end{align*}
This leads us to the Stirling's formula for $n!$. The series $u$ and $v$ play very important role in the proof of Theorem \ref{LWthm} in \cite{LW}, as well as our proof of Conjecture \ref{Conj1} in Sect.\ref{S4}.

\section{Proof of Alexandrov's conjecture}\label{S4}
Let us recall Eq.\eqref{emeq} and \eqref{Alexconj}. Note that in Eq.\eqref{Alexconj} the differential operator acts on $\exp(F_H(u,q))$. Here we re-arrange this equation as the following
\begin{equation}\label{Alexconj2}
C(u)\exp(F_H(u,q))=\exp\left(-\sum_{k=1}^{\infty}\widehat{a}_ku^{k}L_k  \right)\cdot \exp(F_K(q)).
\end{equation}
In this section, we prove Conjecture \ref{Conj1} by showing that 
\begin{equation}\label{ak=em}
\exp\left(-\sum_{k=1}^{\infty}\widehat{a}_ku^{k}L_k  \right)=\exp\left(\sum_{m=1}^{\infty} e_mu^mL_m\right).
\end{equation}
Then, by comparing Eq.\eqref{Alexconj2} with Eq.\eqref{emeq}, we can see that $C(u)$ must be $1$. And to prove the above equation, it is equivalent to show that the following two series are identical
\begin{equation}\label{iden}
\exp(-\sum_{k=1}^{\infty} \widehat{a}_kz^{1-k}\frac{\partial}{\partial z})\cdot z=\exp(\sum_{m=1}^{\infty} e_{m}z^{1-m}\frac{\partial}{\partial z})\cdot z.
\end{equation}
We would like to mention that, in Eq.\eqref{iden}, Proposition \ref{hy} and Lemma \ref{y=K}, the series on the left hand sides of the equations are from Alexandrov's results in \cite{AE} and the series on the right hand sides originate from Theorem \ref{LWthm}. Now let $y(z)$ and $h(z)$ be the inverse function of $f_{+1}$ (Eq.\eqref{f+1}) and $f_{+2}$ (Eq.\eqref{f+2}) respectively. 
Then the function $y(z)$ satisfies
\begin{equation}\label{def:y}
z=\frac{1}{y\exp(y^{-1})\sinh(y^{-1})-1},
\end{equation}
and $h(z)$ is
\begin{equation}\label{h}
h=\left(3z^{-2}\coth(z^{-1})-3z^{-1}\right)^{-\frac{1}{3}}.
\end{equation}
It is clear from Eq.\eqref{f+1+2} that the composition $h(y(z))$ is the inverse function of $f_{+}$, that is,
\begin{equation*}
h(y(z))=\exp(-\sum_{k=1}^{\infty} \widehat{a}_kz^{1-k}\frac{\partial}{\partial z})\cdot z.
\end{equation*}
And the right hand side of Eq.\eqref{iden} is the series $\theta(f(z))$ in Eq.\eqref{def:em}. Therefore, to prove Eq.\eqref{iden}, we need to prove that
\begin{Prop}\label{hy}
$$h(y(z))=\theta(f(z)).$$
\end{Prop}
This proposition can be proved by using the next two lemmas. First we have
\begin{Lem}
For the numbers $\{b_{2i+1}\}$ determined by Eq.\eqref{re}, let
\begin{equation}\label{def:K}
K=\sum_{i=0}^{\infty}b_{2i+1}z^{2i+1}.
\end{equation}
Then
\begin{equation*}
\theta(f(z))=\left(3K^2(f^{-1})\coth K(f^{-1})-3K(f^{-1})\right)^{-\frac{1}{3}}.
\end{equation*}
\end{Lem}
\begin{proof}
For the two series $v$ and $u$ defined by Eq.\eqref{v} and \eqref{u} respectively, we have
\begin{equation*}
K=\frac{1}{2}\left(v-u\right)=\frac{1}{2}\left(W_0(-e^{-1-\frac{z^2}{2}})-W_{-1}(-e^{-1-\frac{z^2}{2}})\right).
\end{equation*}
By Eq.\eqref{def:p} and \eqref{eqn:p}, we have
\begin{equation}\label{Kz}
e^{-1-\frac{z^2}{2} }=K e^{-K\coth K}\csch K.
\end{equation}
Then, we differentiate the above equation by $K$ on both sides. This gives us 
\begin{equation*}
z\frac{\d z}{\d K}=-\frac{1}{K}\left(1-2K\coth K+K^2\csch^2 K \right).
\end{equation*}
By the definition of $K$ in Eq.\eqref{def:K}, we have
$$\int zK \d z=\sum_{i=0}^{\infty}\frac{b_{2i+1}}{2i+3}z^{2i+3}+C.$$
Therefore,
\begin{equation*}
-\int \left(1-2K\coth K+K^2\csch^2 K \right)\d K=\sum_{i=0}^{\infty}\frac{b_{2i+1}}{2i+3}z^{2i+3}+C.
\end{equation*}
We know that
\begin{equation*}
\frac{\d }{\d K}\coth K=-\csch^2 K.
\end{equation*}
Then
\begin{equation*}
\int K^2\csch^2 K \d K=-K^2\coth K+2\int K\coth K\d K.
\end{equation*}
Finally
\begin{align*}
&-\int \left(1-2K\coth K+K^2\csch^2 K \right)\d K  \\
=& -K+K^2\coth K+C.
\end{align*}
In our case, by the definition of $K$ in Eq.\eqref{def:K}, the constant $C$ must be $0$. This implies that
$$K^2\coth K-K=\sum_{i=0}^{\infty}\frac{b_{2i+1}}{2i+3}z^{2i+3}.$$
The lemma follows from the above equation and Eq.\eqref{theta}. 
\end{proof}

By Eq.\eqref{h}, we know
\begin{equation*}
h(y(z))=\left(3y^{-2}\coth(y^{-1})-3y^{-1}\right)^{-\frac{1}{3}}.
\end{equation*}
Therefore, to prove Proposition \ref{hy}, we need the identification between the following two series.
\begin{Lem}\label{y=K}
$$y^{-1}(z)=K(f^{-1}(z)).$$
\end{Lem}
\begin{proof}
First we show that both $y^{-1}(z)$ and $K(f^{-1}(z))$ are solutions of the same equation. From the definition \eqref{functionf} of $f$, we have
$$\frac{1}{1+z^{-1}}e^{-\frac{1}{1+z^{-1}}}=e^{-1- \frac{1}{2}f^{-2}}.$$
And by Eq.\eqref{def:y}, we have
\begin{align*}
\frac{1}{1+z^{-1}}e^{-\frac{1}{1+z^{-1}}} &= y^{-1}e^{-y^{-1}}\csch(y^{-1})\exp\left(-y^{-1}e^{-y^{-1}}\csch(y^{-1})\right)\\
                          &=y^{-1}\exp\left(-y^{-1}(1+e^{-y^{-1}}\csch(y^{-1}))\right)\csch(y^{-1})\\
                          &=y^{-1}e^{-y^{-1}\coth(y^{-1})}\csch(y^{-1}).
\end{align*}
On the other hand, by Eq.\eqref{Kz}, 
$$e^{-1-\frac{1}{2}f^{-2} }=K(f^{-1}) e^{-K(f^{-1})\coth(K(f^{-1}))}\csch (K(f^{-1})).$$
This shows that $y^{-1}(z)$ and $K(f^{-1}(z))$ are both the solution of
\begin{equation*}
\frac{1}{1+z^{-1}}e^{-\frac{1}{1+z^{-1}}}=Fe^{-F\coth(F)}\csch(F).
\end{equation*}

The uniqueness of such solution $F$ follows from the property of $p$ in \eqref{eqn:p}. Note that both $y^{-1}(z)$ and $K(f^{-1}(z))$ are in variables $z^{-1}$. In fact, if we let
$$x=(1+z)^{-1}=\sum_{i=1}^{\infty}(-1)^{i-1}z^{-i},$$
with $|z^{-1}|<1$, then $1-x>0$ and $0<(1-x)e^{-(1-x)}\leq e^{-1}$. By considering
\begin{equation*}
-(1-x)e^{-(1-x)}=-Fe^{-F\coth(F)}\csch(F),
\end{equation*}
we obtain
\begin{equation}\label{eqn:F}
F=\frac{1}{2}\left(W_0(-(1-x)e^{-(1-x)})-W_{-1}(-(1-x)e^{-(1-x)}) \right).
\end{equation}
Since we require the coefficient of $z^{-1}$ in the series $F$ to be $1$, from Eq.\eqref{eqn:mu} and its properties introduced in Sect.\ref{S3}, we require that $x>0$. Hence
\begin{equation}\label{def:F}
F=\frac{1}{2}(-1+x+1+\mu)=x+\frac{1}{2}\sum_{n=2}^{\infty}c_nx^n,
\end{equation}
where $c_n$ are determined by Eq.\eqref{re:c}. And the proof is completed.
\end{proof}

Finally, as a completion of proving Conjecture \ref{Conj1}, we show that the function
\begin{equation}\label{def:H}
H=(3F^2\coth F-3F)^{-\frac{1}{3}}
\end{equation}
satisfies
\begin{equation}\label{eqn:E}
(1-x)e^{-(1-x)}=Ee^{-E},
\end{equation}
where
$$E=1+\sqrt{x^2+\frac{4}{3H^3}}.$$
Eq.\eqref{eqn:E} corresponds to Eq.\eqref{f+} with $x=(1+f_{+})^{-1}$ and $H=z$. For Eq.\eqref{eqn:F}, here we write $F=\frac{1}{2}(W_0-W_{-1})$ for simplicity. By Eq.\eqref{eqn:F}, \eqref{def:F} and \eqref{def:H}, we have
\begin{align*}
H &= \left(3F^2\frac{\exp(W_0-W_{-1})+1}{\exp(W_0-W_{-1})-1}-3F\right)^{-\frac{1}{3}}\\
&= \left(3F^2\frac{W_{-1}+W_0}{W_{-1}-W_0}-3F\right)^{-\frac{1}{3}}\\
&=\left(\frac{3}{4}(\sum_{n=1}^{\infty}c_nx^n)^2-\frac{3}{4}x^2\right)^{-\frac{1}{3}}.
\end{align*}
This implies that $\sqrt{x^2+\frac{4}{3H^3}}$ satisfies Eq.\eqref{eqn:mu}. Hence $H$ satisfies Eq.\eqref{eqn:E}. And we have completed the proof of Conjecture \ref{Conj1}.

\vspace{20pt}

\noindent
{\bf Remark:}  In \cite{AE} the following series was considered in the Kac-Schwarz description of the Hodge tau-function:
\begin{equation*}
\eta=\frac{e^q\sinh(q)}{q}-1.
\end{equation*}
And $q$ can be expressed as
\begin{equation*}
q=\frac{S_{-}-S_{+}}{2}, 
\end{equation*}
where $S_{\pm}$ are the two solutions of the equation
\begin{equation*}
Se^{-S}=\frac{1}{1+\eta}e^{-\frac{1}{1+\eta}}.
\end{equation*}
Here $S_{+}$ represents $(1+\eta)^{-1}=\sum_{k=0}^{\infty}(-1)^k\eta^k$. The function $y$ in Eq.\eqref{def:y} is defined as the inverse function of $f_{+1}$. The relation between them are $\eta=z^{-1}$ and $q=y^{-1}$. On the other hand, the function $p$ in Eq.\eqref{eqn:p} is defined as a branch difference of the Lambert W function, where $W_0(z)$, usually called as the principal branch, can be expressed as the Taylor series
$$W_0(z)=\sum_{n=1}^{\infty}\frac{(-n)^{n-1}}{n!}z^n.$$
The series $S$ and $W$ are clearly different, but they are related as
\begin{equation*}
S(\eta)=-W\left(-\frac{1}{1+\eta}e^{-\frac{1}{1+\eta}}\right).
\end{equation*}
We can see that $S$ can also be seen as a series for the Lambert W function. 

\section*{Acknowledgement}

The author would like to thank the referee(s) for comments and suggestions in improving this paper's presentation. Research of the author is supported by the National Natural Science Foundation of China (Grant No.11701587).

\vspace{10pt} \noindent
\\
\footnotesize{\sc gehao wang: 
school of mathematics (zhuhai), sun yat-sen university, zhuhai, china. }\\
\footnotesize{E-mail address:  gehao\_wang@hotmail.com}

\end{document}